\documentclass[11pt]{article}

\usepackage{xargs}[2008/03/08]
\usepackage{amsmath}
\usepackage{amsthm}
\usepackage{amssymb}
\usepackage{stackrel}
\usepackage{mathtools}
\usepackage{enumitem}
\usepackage{hyperref}
\usepackage[left=3.5cm, right=3.5cm]{geometry}
\mathtoolsset{showonlyrefs}
\usepackage{dsfont}

\title{Limits for Rumor Spreading in stochastic populations\footnote{This work has received funding from the European
    Research Council (ERC) under the European Union's Horizon 2020 research and
    innovation programme (grant agreement No 648032).}}
\usepackage{authblk}

\author[a]{Lucas Boczkowski}
\author[c]{Ofer Feinerman}
\author[a]{Amos Korman}
\author[b]{Emanuele Natale}

\affil[a]{CNRS, IRIF, Université Paris Diderot, 75013 Paris, France\\
\texttt{first.last@irif.fr}}
\affil[b]{Max-Planck-Institut für Informatik, 66123 Saarbrücken, Germany\\
\texttt{enatale@mpi-inf.mpg.de}}
\affil[c]{Weizmann Institute of Science, 76100 Rehovot, Israel\\
\texttt{feinermanofer@gmail.com}}

\theoremstyle{plain}
\newtheorem{theorem}{Theorem}[section]

\newtheorem{claim}[theorem]{Claim}
\newtheorem{lemma}[theorem]{Lemma}

\newtheorem{corollary}[theorem]{Corollary}

\newtheorem{remark}[theorem]{Remark}

\theoremstyle{definition}
\newtheorem{definition}[theorem]{Definition}

\usepackage{bbm}
\usepackage{xspace}

\newcommand{\pull}{\ensuremath{\mathcal{PULL}}\xspace}

\newcommand{\push}{\ensuremath{\mathcal{PUSH}}\xspace}
\newcommand{\spull}{\ensuremath{\mathcal{\mbox{\em sequential}\textnormal{-}PULL}}\xspace}
\newcommand{\ppull}{\ensuremath{\mathcal{\mbox{\em parallel}\textnormal{-}PULL}}\xspace}
\newcommand{\ppullk}{\ensuremath{\mathcal{\mbox{\em parallel}\textnormal{-}PULL}}(k)\xspace}

\newcommand{\AVDT}{\tt ACDT}

\renewcommand{\epsilon}{\varepsilon}

\newcommand{\kl}{KL}
\newcommand{\bias}{s}
\newcommand{\type}{\eta}
\newcommand{\epsnorm}{d_{\epsilon}(x^{(<t)})}

\renewcommand{\Pr}{P}

\global\long\def\cond{\ |\ }

\renewcommand{\epsilon}{\varepsilon}
\newcommand{\bx}{\mathbf{x}}
\newcommand\bigO{\mathcal{O}}

\newcommand{\Natural}{\mathbb{N}}
\newcommandx\seqobs[1][usedefault, addprefix=\global, 1=t]{\bx^{(\leq#1)}}
\newcommandx\tildeseqobs[1][usedefault, addprefix=\global, 1=t]{{\tilde \bx}^{(\leq#1)}}
\newcommandx\Seqobs[1][usedefault, addprefix=\global, 1=t]{\mathbf{X}^{(\leq#1)}}
\newcommandx\tildeSeqobs[1][usedefault, addprefix=\global, 1=t]{\mathbf{\tilde X}^{(\leq#1)}}
\global\long\def\rpull{\ensuremath{\mathcal{\mbox{\em broadcast}\textnormal{-}PULL}}\xspace}

\newcommand{\obs}[1]{x^{(#1)}}

\begin{document}
\date{}

\begin{titlepage}
\newgeometry{top=4cm}
\maketitle

\date{}

\def\thefootnote{\fnsymbol{footnote}}

\vspace{0.5cm}
\begin{abstract}\parskip0.08cm
    Biological systems can share and collectively process
    information to yield emergent effects, despite inherent noise in
    communication. While man-made systems often employ intricate structural
    solutions to overcome noise, the structure of many biological systems is
    more amorphous. It is not well understood how communication noise may
    affect the computational repertoire of such groups. To approach this
    question we consider the basic collective task of rumor spreading, 
    in which information from  few knowledgeable sources must reliably flow into the rest of the population. 

    In order to study the effect of communication noise on the 
    ability of groups that lack stable structures to efficiently solve this
    task, we consider a noisy version of the uniform $\mathcal{PULL}$ model.
    We prove a lower bound which implies that, in the presence of even moderate
    levels of noise that affect all facets of the communication, no scheme can
    significantly outperform the trivial
    one in which agents have to wait until directly
    interacting with the sources. 
    Our results thus
    show an exponential separation between the uniform $\mathcal{PUSH}$ and
    $\mathcal{PULL}$
    communication models in the presence of noise. 
    Such separation may be interpreted as suggesting that, in order to achieve efficient rumor
    spreading, a system must exhibit either some degree of structural stability or,
    alternatively, some facet of the communication which is immune to noise. 

    We corroborate our theoretical findings with a new analysis of experimental data
    regarding recruitment  in {\em Cataglyphis niger} desert ants. 
 \end{abstract}

\end{titlepage}
\section{Introduction}
\subsection{Background and motivation}

Systems composed of tiny mobile components must function under conditions of unreliability. In particular, any sharing of information is inevitability subject to communication noise.  The effects of communication noise in distributed living systems appears to be highly variable. 
While some systems disseminate information efficiently and
reliably despite communication noise \cite{abeles1991corticonics,feinerman2008reliable,cavagna2010scale,marras2012information,rosenthal2015revealing},  others generally refrain from acquiring social information, consequently
losing all its potential benefits
\cite{giraldeau2002potential,rieucau2009persuasive,templeton1995patch}. It is
not well understood which characteristics of a distributed system are crucial
in facilitating noise reduction strategies and, conversely, in which systems
such strategies are bound to fail. Progress in this direction may be valuable
towards better understanding  the constraints that govern the evolution of
cooperative biological systems.

Computation under noise has been extensively studied in the  
 computer science community. These studies suggest that different forms of error correction ({\em e.g.,} redundancy) are highly useful in maintaining reliability despite noise
\cite{Alon16,gamal11,XR17,Von56}. All these, however, require the ability to
transfer significant amount of information over stable communication channels.
Similar redundancy methods may seem biologically plausible in systems
that enjoy stable structures, such as brain tissues.

The impact of noise in stochastic systems with ephemeral connectivity patterns
is far less understood. To study these, we focus on  {\em rumor spreading} - a
fundamental information dissemination task that is a prerequisite to almost
any distributed system \cite{SODA,CHHKM12,DGH88,KSSV00}. A
successful and efficient rumor spreading process is one in which a large
group manages to quickly learn information initially held by one or a few
informed individuals. 
Fast information flow to the whole
group dictates that messages be relayed between individuals. Similar to the
game of Chinese Whispers, this may potentially result in runaway buildup of
noise and loss of any initial information \cite{Cascades}. It currently remains
unclear what are the precise conditions that enable fast rumor spreading. On
the one hand,  recent works indicate that in some models of
random noisy interactions, a collective coordinated process can in fact
achieve fast information spreading \cite{OHK14,NF16}. These models, however,
are based on {\em push} operations that inherently include a certain reliable
component (see more details in Section \ref{sec:separation}). On the other
hand, 
other works consider computation through noisy operations, and
show that several distributed tasks 
require significant running time
\cite{GoyalKS08}. The tasks considered in these works (including the problem of
learning the input bits of all processors, or computing the parity of all the inputs) were motivated by 
computer applications, and may be less relevant for biological contexts. Moreover, they appear to be more demanding 
than basic  tasks, such as rumor spreading, and hence it is unclear how to relate 
bounds on the former problems to the latter ones.

In this paper we take a general stance to identify limitations under which
reliable and fast rumor spreading cannot be achieved. Modeling a well-mixed
population, we consider a passive communication scheme in which information flow occurs  as one agent observes the cues displayed by another. If these interactions are perfectly reliable,  the population could achieve extremely fast rumor spreading \cite{KSSV00}. In contrast, here we focus on
the situation in which messages are noisy.
Informally, our main theoretical result states that when all components of
communication are noisy then fast rumor spreading through large
populations is not feasible. 
In other words, 
our results imply that fast rumor spreading can only be
achieved if either 1) the system exhibits some degree of structural stability or
2) some facet of the pairwise communication is immune to noise.
In fact, our lower bounds hold even when individuals are granted
unlimited computational power and even when the system can take advantage of
complete synchronization. 

Finally, we corroborate our theoretical findings with new analyses regarding the efficiency of information dissemination during recruitment by desert ants. More specifically, we analyze data from an experiment conducted at the Weizmann Institute of Science, concerning 
recruitment  in {\em Cataglyphis niger} desert ants \cite{Razin}. 
These analyses suggest that this biological system lacks reliability in all its communication components, and
 its deficient performances qualitatively validate our predictions. We stress that this part of the paper is highly uncommon. 
Indeed, using empirical biological data to validate predictions from theoretical distributed computing is extremely rare. 
We believe, however, that this interdisciplinary methodology may carry significant potential, and hope that this paper could be useful for future works that will follow this framework.

\subsection{The problem}
\label{sec:problem}
In this section, we give an intuitive description of the problem. Formal definitions are provided in Section \ref{sec:model}.

Consider a population of $n$ {\em agents}. Thought of as computing entities,
assume that each agent  has a discrete internal {\em state},
and can execute randomized algorithms - by internally flipping coins. In addition, each agent has an {\em opinion}, which we assume for simplicity to be binary, {\em i.e.,} either 0 or 1.
A small number, $s$, of agents play the role of {\em sources}. Source agents are aware of their role and share the same opinion, referred to as the {\em correct opinion}. The goal of all agents is to have their opinion coincide with the correct opinion. 

To achieve this goal,  each agent continuously displays one of several {\em
messages} taken from some finite alphabet $\Sigma$. Agents interact  according to a random pattern, termed as the
{\em parallel-$\pull$} model: In each round $t\in \Natural^+$, each agent $u$
observes the message currently displayed by another agent $v$, chosen uniformly at random (u.a.r) from all agents. Importantly, communication is noisy, hence the message observed
by $u$ may differ from that displayed by $v$. The noise is characterized by a
\emph{noise parameter}  $\delta>0$. Our model encapsulates a large family of noise
distributions, making our bounds highly general. Specifically, the noise
distribution  can take {\em any} form, as long as it satisfies the following
criterion.

\begin{definition}[The $\delta$-uniform noise criterion]
    \label{def:noisecrit}
    Any time some agent $u$ observes an agent $v$ holding some message $m \in
    \Sigma$, the probability that $u$ actually receives a message $m'$ is at
    least $\delta$, for any $m' \in \Sigma$. All noisy samples
    are independent.
\end{definition}

When messages are noiseless, it is easy to see that the number of rounds that are required to guarantee that all agents hold the correct opinion with high probability
is $\bigO(\log n)$ \cite{KSSV00}. 
In what follows, we aim to show that when the $\delta$-uniform noise criterion
is satisfied, the number of rounds required until even one non-source agent can
be moderately certain about the value of the correct opinion is very large.
Specifically, thinking of $\delta$ and $\bias$ as constants independent of the
population size $n$, this time is at least $\Omega(n)$. 

To prove the lower bound, we will bestow the agents with  capabilities that far
surpass those that are reasonable for biological entities. These include: 
\begin{itemize}
    \item Unique identities: Agents have unique identities in the range
        $\{1,2,\ldots n\}$. When observing agent $v$, its identity is received
        without noise. 
    \item Complete knowledge of the system: Agents have access to all  
        parameters of the system (including $n$, $s$, and $\delta$) as well as to the full knowledge
        of the initial configuration except, of course, 
        the correct opinion and the identity of the sources.
        In addition, agents have access to  the results of random coin
        flips used internally by all other agents.
    \item Full synchronization: Agents know when the execution starts, and can
        count rounds. 
\end{itemize}
We  show that even given this extra computational power, fast convergence
cannot be achieved. 

\subsection{Our contributions}\label{sec:results}
\subsubsection{Theoretical results}\label{sec:theoretical}
In all the statements that follow we consider the parallel-$\pull$ model satisfying the $\delta$-uniform noise criterion, where  $cs/n< \delta\leq 1/2$ for some sufficiently large constant $c$.
Note that our criterion given in Definition \ref{def:noisecrit} implies that
$\delta\leq 1/|\Sigma|$. 
Hence, the previous lower bound on $\delta$ implies a restriction on the
alphabet size, specifically, $|\Sigma|\leq n/(cs)$.

\begin{theorem}
    \label{thm:ppullk}
    Any rumor spreading protocol cannot converge in less than
    ${\Omega}(\frac{n\delta}{\bias^2(1-2\delta)^2})$ rounds. 
\end{theorem}
Recall that  a source is aware that it is a source, but if it wishes to identify itself as such to agents that observe it, it must encode this information in a
message, which is, in turn, subject to noise. We also consider the case in
which
an agent can reliably identify a source when it observes one (i.e.,, this
information is not noisy). For this case, the following  bound, which is
weaker than the previous one but still polynomial, applies (a formal proof appears in Section \ref{app:cor_source}): 

\begin{corollary}
    \label{cor:informal_reliable_source}
    Assume that sources are reliably detectable. 
    There is no rumor spreading protocol that
    converges in less than
    ${\Omega}((\frac{n\delta}{s^2(1-2\delta)^2})^{1/3})$ rounds.   
\end{corollary}

Our results suggest that, in contrast to systems that enjoy stable connectivity, structureless systems are highly sensitive to
communication noise. More concretely, the two crucial assumptions that make our lower bounds work are: 
1) stochastic interactions, and 
2)  $\delta$-uniform noise (see the right column of Figure \ref{fig:non-uniform}).
When agents can stabilize their interactions  the first assumption is violated.
In such cases, agents can overcome noise by employing simple error-correction
techniques, {\em e.g.,} using redundant messaging  or waiting for
acknowledgment before proceeding.  As demonstrated in Figure
\ref{fig:non-uniform} (left column), when the noise is not uniform, it might be possible to overcome it with simple techniques 
based on  using default neutral messages, and employing exceptional
distinguishable signals only when necessary.

\begin{figure}[h!]
\begin{center}
\includegraphics[width=10cm]{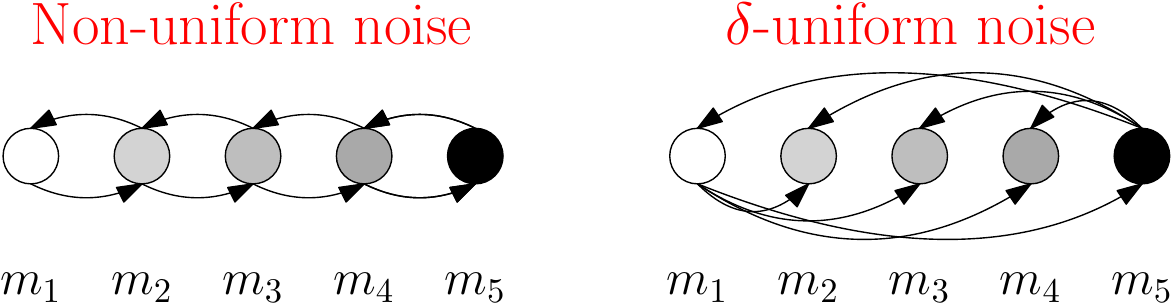}
    \caption{\small {\bf Non-uniform noise vs. uniform noise.} On the left, we
    consider an example with non-uniform noise. Assume that the message
    vocabulary consists of 5 symbols, that is,
    $\Sigma=\{m_1,m_2,m_3,m_4,m_5\}$, where $m_1=0$ and $m_5=1$, represent the
    opinions. Assume that noise can occur only between consecutive messages.
    For example, $m_2$ can be observed as either $m_2$, $m_3$ or  $m_1$, all
    with positive constant probability, but can never be viewed as $m_4$ or
    $m_5$. 
    In this scenario, the population can quickly converge on the correct
    opinion by executing the following. The sources always display the
    correct opinion, {\em i.e.,} either $m_1$ or $m_5$, and each other agent displays
    $m_3$ unless it has seen either $m_1$ or $m_5$ in which case it  adopts
    the opinion it saw and displays it. In other words, $m_3$ serves as a
    default message for non-source agents, and $m_1$ and $m_5$ serve as
    attracting sinks. 
    It is easy to see that the correct opinion will propagate quickly through
    the system without disturbance, and within $\bigO(\log n)$ number of
    rounds, where $n$ is the size of the population, all agents will hold it
    with high probability. In contrast, as depicted on the right picture, if every
    message  can be observed as any other message with some constant positive
    probability  (for clarity, some of the arrows have been omitted from the sketch), then convergence cannot
    be achieved in less than $\Omega(n)$ rounds, as Theorem \ref{thm:ppullk} dictates.}
    \label{fig:non-uniform}
\end{center}
\end{figure}

\subsubsection{Exponential separation between \push and \pull}
\label{sec:separation}

Our lower bounds on the parallel-\pull model (where agents observe other
agents) should be contrasted with known results in the parallel-\push model, 
which is the push equivalent to parallel-\pull model, where in each round each
agent may or may not actively push a message to another agent chosen u.a.r. (see also
Section \ref{sec:random-inter}). Although never proved,
and although their combination is known to achieve more power than each of them
separately \cite{KSSV00}, researchers often view the  parallel-\pull and parallel-\push models as
very similar on complete communication topologies. Our lower bound result, however, undermines this
belief, proving that in the context of noisy communication, there is an
exponential separation between the two models. Indeed, when the noise level is
constant for instance,  convergence (and in fact, a much stronger convergence
than we consider here) can be achieved in the  parallel-\push using only
logarithmic number of rounds \cite{OHK14,NF16}, 
by a simple strategy composed of two stages. 
The first stage consists of providing all agents with a guess about the source's
opinion, in such a way that ensures a non-negligible bias toward the
correct guess. 
The second stage then boosts this bias by progressively amplifying it.  
A crucial aspect in the first stage is that agents remain silent until a certain
point in time that they start sending messages continuously, which happens
after being contacted for the first time.  
This prevents agents from starting to spread information before they have
sufficiently reliable knowledge.  
It further allows to control the dynamics of the information spread in a balanced
manner. More specifically, marking an edge corresponding to a message received
for the first time by a node, the set of marked edges forms a spanning tree of
low depth, rooted at the source.
The depth of such tree can be interpreted as the deterioration of the message's
reliability. 
%

On the other hand, as shown here,  in the parallel-\pull model, even with the
synchonization assumption,  rumor spreading cannot be achieved in
less than a linear number of rounds. 
Perhaps the main reason why these two models are often considered similar is
that with an extra bit in the message, a \push protocol can be
\emph{approximated} in the \pull model, by letting this bit indicate whether
the agent in the \push model was aiming to push its message. However, for such
a strategy to work, this extra bit has to be reliable. Yet, in the noisy \pull model, no bit is safe from noise, and hence, as we show, such
an approximation cannot work. In this sense, the extra power that the noisy \push
model gains over the noisy \pull model, is that the very fact that one node
attempts to communicate with another is reliable. 
This, seemingly
minor, difference carries significant consequences.

\subsubsection{Generalizations}
Several of the assumptions discussed earlier for the parallel-\pull model were
made for the sake of simplicity of presentation. In fact, our results can be
shown to hold under more general conditions, that include: 1) different rate
for sampling a source, and 
2) a more relaxed noise criterion. 
In addition, our theorems were stated with respect to the parallel-\pull model. In this model, at every round, each agent samples a single agent u.a.r. In fact, for any integer $k$, our analysis can be applied to the model in which, at every round, each agent observes $k$ agents chosen u.a.r. In this case, the lower bound would simply reduce by a factor of $k$. Our analysis can also apply to a sequential variant, in which in each time step, two agents $u$ and $v$ are chosen u.a.r from the population and $u$ observes $v$. In this case, our lower bounds would multiply by a factor of $n$, yielding, for example, a lower bound of $\Omega(n^2)$ in the case where $\delta$ and $s$ are constants\footnote{This increase in not surprising as each round in the parallel-\pull model consists of $n$ observations, while the sequential model consists of only one observation in each time step.}.
\subsubsection{Recruitment in desert ants}

Our theoretical results assert that efficient rumor spreading  in large groups could not be achieved without some degree of communication reliability. 
An example of a biological system whose communication reliability appears to be
deficient in all of its components is recruitment in {\em Cataglyphis niger}
desert ants. In this species, when a forager locates an oversized food
item, she returns to the nest to recruit other
ants to help in its retrieval \cite{Amor,Razin}. 

We complement our theoretical findings by providing new analyses from an
experiment on this system conducted at the Weizmann Institute of Science \cite{Razin}. 
In such experimental setting, we interpret our theoretical findings as an
abstraction of the interaction modes between ants. While such high-level
approximation may be considered very crude, we retain that it constitutes a
good trade-off between analytical tractability and experimental data.

In our experimental setup recruitment happens in the
small area of the nest's entrance chamber (Figure \ref{fig:razin}a). We find that within this confined
area, the interactions between ants are nearly uniform \cite{Ant-estimation}, 
such that an ant cannot control which of her nest mates she meets next (see Figure \ref{fig:razin}b).  
This random meeting pattern coincides with the first main assumption of our model. Additionally, it has been shown that 
 recruitment in \emph{Cataglyphis niger} ants relies on rudimentary alerting interactions
\cite{ANNA,Bert} which are subject to high levels of noise \cite{Razin}. 
Furthermore, the responses to a
recruiting ant and to an ant that is randomly moving in the nest are extremely
similar \cite{Razin}. 
Although this may resemble a noisy push interaction scheme, ants cannot
reliably distinguish an ant that attempts to transmit information from any
other non-communicating individual. In our theoretical framework, the latter
fact means that the structure of communication is captured by a noisy-pull
scheme (see more details about $\push$ vs.~$\pull$ in Section
\ref{sec:separation}).

\begin{figure}[h!]
\begin{center}
    \includegraphics[width=8cm]{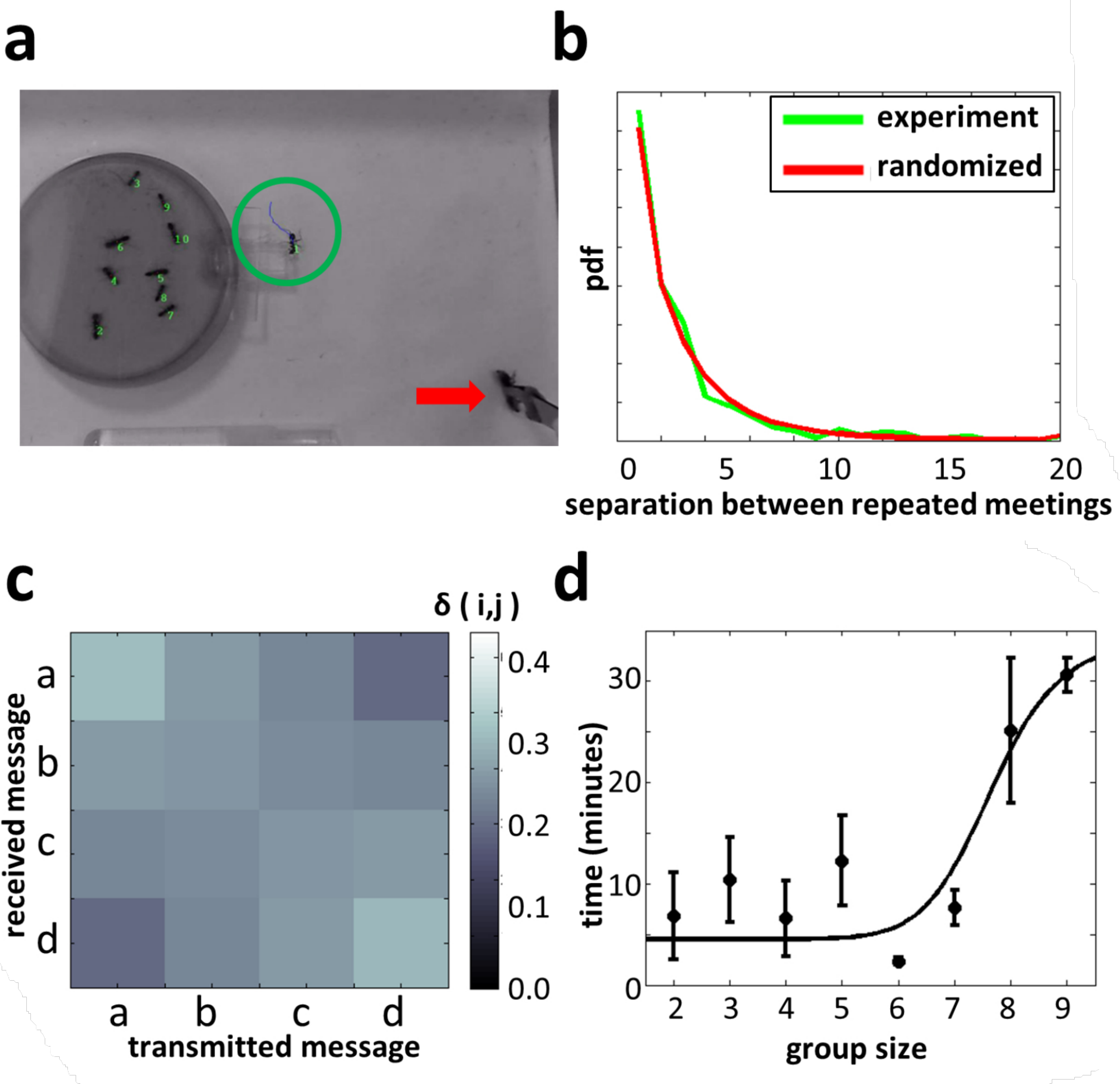}
    \caption{\small {\bf Unreliable communication and slow recruitment by desert ant
        (\emph{Cataglyphis niger}).} \textbf{a.} The experimental setup. The
        recruiter ant (circled) returns to the nest's entrance chamber (dark,
        9cm diameter, disc) after finding the immobilized food item (arrow).
        Group size is ten. \textbf{b.} A \textit{pdf} of the number of
        interactions that
        an ant experiences before meeting the same ant twice. The \textit{pdf} is
        compared to uniform randomized interaction pattern. Data summarizes $N=671$
        interactions from seven experiments with a group size of 6 ants.
        \textbf{c.} 
        Interactions with moving ants where classified into four different
        messages ('a' to 'd') depending on the ants' speed. The noise at which
        messages were confused with each other was estimated according to the
        response recipient, initially stationary, ants (see Materials and Methods). Gray
        scale indicates the estimated overlap between every two messages
        $\delta(i,j)$. Note that $\delta=\min(\delta(i,j)) \approx 0.2 $. Data
        collected over $N=64$ interactions. \textbf{d.} The mean time it takes an ant that is informed about the food to recruit two nest-mates to exit the nest is presented for two group size ranges.}
    \label{fig:razin}
\end{center}
\end{figure}

 It has
previously been shown that the information an ant passes in an interaction can
be attributed solely to her speed before the interaction \cite{Razin}. Binning
ant speeds into four arbitrary discrete messages and measuring the responses of
stationary ants to these messages, we can estimate the probabilities of one
message to be mistakenly perceived as another one (see Materials and Methods).
Indeed, we find that this communication is extremely noisy and complies with
the uniform-noise assumption with a $\delta$ of approximately $0.2$ (Figure
\ref{fig:razin}c).

Given the coincidence between the communication patterns in this ant system and
the requirements of our lower bound we expect long delays before any uninformed
ant can be relatively certain that a recruitment process is occurring. We
therefore measured the time it takes an ant, that has been at the food source,
to recruit the help of two nest-mates. We find that this time increases with
group size ($p<0.05$ Kolmogorov-Smirnov test over $N=24$ experiments, Figure \ref{fig:razin}d). Thus, in this system, inherently noisy interactions on
the microscopic level have direct implications on group level performance. While group sizes in these experiments are small, we nevertheless find these recruitment times in accordance with our asymptotic theoretical results.
More details on the experimental methodology can be found in Appendix \ref{app:exp}. 

\subsection{Related work}

Lower bound approaches in biological contexts are still extremely rare
\cite{bialek,feinerman2013theoretical}.
Our approach can be framed within the general endeavour of addressing problems
in theoretical biology through the algorithmic perspective of theoretical
computer science \cite{chazelle_natural_2009, chastain_algorithms_2014}.

The computational study of abstract systems composed of simple individuals that
interact using highly restricted and stochastic interactions has recently been
gaining considerable attention in the community of theoretical computer
science. Popular models include {\em population protocols}
\cite{Pop2}, which typically consider constant  size
individuals that interact in pairs (using constant size messages) in random
communication patterns, and the {\em beeping} model \cite{beeping1}, which
assumes a fixed network with extremely restricted communication. Our model
also falls in this framework as we consider the $\pull$ model \cite{DGH88,
KSSV00, KDG03} with constant size messages. So far, despite interesting works
that consider different fault-tolerant contexts \cite{Aspnes,Pop-ss},
most of the progress in this framework considered noiseless scenarios. 

In {\em Rumor Spreading} problems  (also referred to as {\em Broadcast}) a
piece of information typically held by a single designated agent is to be
disseminated to the rest of the population.
It is the subject of a vast literature in theoretical computer science, and
more specifically in the distributed computing community, see, {\em e.g.,}
\cite{SODA,CHHKM12,DGH88, DGM11,  OHK14,  GoyalKS08, KSSV00, Pittel}. 
While some works assume a fixed topology, the canonical setting does not assume
a network. Instead agents communicate through uniform $\push / \pull$ based
interactions (including the {\em phone call} model), in which agents interact in pairs with other agents independently chosen at each time step uniformly at random from all agents in the population. The success of such
protocols is largely due to their inherent simplicity and fault-tolerant
resilience \cite{ES09,KSSV00}.
In particular, it has been shown that under the $\push$ model, there exist
efficient rumor spreading protocol that uses a single bit per message and can
overcome flips in messages (noise) \cite{OHK14}. 

The line of research initiated by El-Gamal \cite{gamalOriginal}, also studies a
broadcast problem with noisy interactions. The regime however is rather
different from ours: all $n$ agents hold a bit they wish to transmit to a
single receiver. 
This line of research culminated in the $\Omega(n \log \log n)$  lower bound on
the number of messages shown in \cite{GoyalKS08}, matching the upper bound
shown many years sooner in \cite{Gal}.

\section{Formal description of the models}
\label{sec:model}

We consider a population of $n$ agents that interact stochastically and aim to converge on a particular opinion held by few knowledgable individuals. For simplicity, 
we assume that the set of opinions contain two opinions only, namely,  0 and 1.

As detailed in this section, we shall assume that agents have access to significant amount of resources, often exceeding reasonable more realistic assumptions.
Since we are concerned with lower bounds, we do not loose generality from such permissive assumptions.  
These liberal assumptions  will actually simplify our proofs. One of these assumptions is the assumption that each agent is equipped with a unique identity $id(v)$ in the range $\{1,2,\ldots,n\}$ (see more details in Section \ref{sec:liberal}). 

\subsection{Initial configuration}
The initial configuration is described  in several layers. First, the {\em neutral initial configuration} corresponds to the initial states of the agents, before the sources and the desired opinion to converge to are set. Then,  
 a random initialization is applied to the given neutral initial configuration, which determines the set of sources and the opinion that agents need to converge to. This will result in what we call the {\em charged initial configuration}. It can represent, for example, an external event that was identified by few agents which now need to deliver their knowledge to the rest of the population.
 
\textbf{Neutral Initial Configuration $\bx^{(0)}$.} Each agent $v$ starts the
execution with an \emph{input} that contains, in addition to its identity,
    an initial {\em state} taken from some discrete set of states,  and\footnote{The opinion of an agent could have been considered as part of the state of the agent. We separate these two notions merely for the presentation purposes.}
    a binary {\em opinion} variable $\lambda_v\in \{0,1\}$.
The {\em neutral initial configuration} $\bx^{(0)}$ is the vector whose $i$'th index,
$\bx^{(0)}_i$ for $i\in [n]$, is the input of the agent with
identity $i$. 

\textbf{Charged Initial Configuration and  Correct Opinion.} 
The charged initial configuration is determined in three stages. The first
corresponds to the random selection of sources,  the second to the selection of
the correct opinion, and the third to a possible update of states of sources,
as a result of being selected as sources with a particular opinion. 

\begin{itemize}[noitemsep,topsep=0pt]

    \item \textbf{1st stage - Random selection of sources.}
        Given an integer $s\leq n$, a set $S$ of size $s$ is chosen uniformly at random
        (u.a.r) among the agents. The agents in $S$ are called {\em sources}.  Note
        that any agent has equal probability of being a  source. We assume that each
        source knows it is a source, and conversely, each non-source  knows it is not a
        source. 

    \item \textbf{2nd stage - Random selection of correct opinion.}
        In the main model we consider, after sources have been determined in the first
        stage, the sources are randomly initialized with an opinion, called the {\em
        correct opinion}. That is, a fair coin is flipped to determine an opinion in
        $\{0,1\}$ and all sources are assigned with this opinion. 

    \item \textbf{3rd stage - Update of initial states of sources.}
        To capture a change in behavior as a result of being selected as a source with
        a particular opinion, we assume that once the opinion of a source $u$ has been
        determined, the initial state of $u$ may change according to some distribution
        $f_{source-state}$ that depends on (1) its identity,  (2) its opinion, and (3)
        the neutral configuration. Each source samples its new state independently. 
\end{itemize}

\subsection{Alphabet and noisy messages}
Agents communicate by observing each other according to some random pattern (for details see Section \ref{sec:random-inter}). To improve communication agents may choose which content, called {\em message}, they wish to reveal to other agents that observe them. Importantly, however, such messages are subject to noise. 
More specifically, at any given time, each agent $v$ (including
sources) displays a message $m\in \Sigma$, where $\Sigma$ is some finite
alphabet.  The alphabet $\Sigma$ agents use to communicate may be
richer than the actual information content they seek to disseminate, namely, their opinions. This, for
instance, gives them the possibility to express several levels of certainty
\cite{PLOS}. We can safely assume that the size of $\Sigma$ is at least two, and that 
$\Sigma$ includes both  symbols $0$ and $1$.  We are mostly concerned with the case where $\Sigma$ is of constant size ({\em i.e.,} independent of the number of agents), but note that our results hold 
for any size of the alphabet $\Sigma$, as long as the noise criterion is satisfied (see below).

\textbf{$\delta$-uniform noise.}
When an agent $u$ {\em observes} some agent $v$, it  receives a
sample of the message currently held by $v$. 
More precisely, for any $m, m' \in \Sigma$, let $P_{m, m'}$ be the probability
that, any time some agent $u$ observes an agent $v$ holding some message $m \in
\Sigma$, $u$ actually receives message $m'$. 
The probabilities $P_{m, m'}$ define the entries of the noise-matrix $P$
\cite{NF16}, which does not depend on time. 
We hereby also emphasize that the agents' samples are independent.

The noise in the sample is
characterized by a \emph{noise parameter} $0<\delta\leq 1/2$.  One of the important aspects in our theorems is that they are general enough to hold assuming {\em any} distribution governing the noise, as long as it satisfies the following noise criterion. 
\begin{definition}[The noise ellipticity parameter $\delta$]
    We say that the noise has
    ellipticity $\delta$ if $P_{m, m'} \geq \delta$ for any $m, m' \in \Sigma$. 
\end{definition}
Observe that the aforementioned criterion implies that $\delta\leq 1/|\Sigma|$, and that the case $\delta = 1/|\Sigma|$ corresponds to messages being completely random, and the rumor spreading
problem is thus unsolvable. 
We next define a weaker criterion,  that is particularly meaningful in cases in which sources are more restricted in their message repertoire than general agents. 
This may be the case, for example, if sources always choose to display their opinion as their message (possibly together with some extra symbol indicating that they are sources). Formally, we define $\Sigma'\subseteq \Sigma$ as the set  of possible messages that a source can hold together with the set of messages that can be observed when viewing a source ({\em i.e.,} after noise is applied). Our theorems actually apply to the following criterion, that requires that only messages in $\Sigma'$ are attained due to noise with some sufficient probability.
\begin{definition}[The relaxed noise ellipticity parameter $\delta$]
    We say that the noise has $\Sigma'$-relaxed ellipticity $\delta$ if $P_{m, m'} \geq \delta$
    for any $m\in \Sigma$ and $ m' \in \Sigma'$. 
\end{definition}


\subsection{Random interaction patterns}
\label{sec:random-inter} 

We consider  several basic interaction patterns. 
Our main model is the \emph{parallel-\pull} model. In this model,  time
is divided into \emph{rounds}, where at each round $i\in \Natural^+$, each agent $u$ independently
selects an agent $v$ (possibly $u=v$) u.a.r from the population and then $u$ observes the message held by $v$.
The \emph{parallel-\pull} model should be contrasted with the
\emph{parallel-\push} model, in which $u$
can choose between \emph{sending} a message to the selected node $v$ or doing nothing. 
We shall also consider the following variants of \emph{\pull} model.

\begin{itemize}[noitemsep,topsep=0pt]
    \item \ppullk. Generalizing \ppull for an integer $1\leq k\leq n$, the \ppullk model allows agents to observe $k$ other agents in each round. That is, at each round $i\in \Natural^+$, each agent independently selects a set of $k$ agents (possibly including itself) u.a.r from the population and observes each of them. 
    \item  \spull. In each time step $t\in \Natural^+$,
         two agents $u$ and $v$ are selected uniformly at random (u.a.r)
        among the population, and agent $u$ observes $v$. 
    \item \rpull. It each time step $t\in \Natural^+$ one
        agent is chosen u.a.r. from the population and all agents observe it, receiving the same
        noisy sample of its message\footnote{The \rpull model is mainly used for technical considerations. We use it in our
proofs as it simplifies our arguments while not harming their generality.
Nevertheless, this broadcast model can also capture some situations in which
agents can be seen simultaneously by many other agents, where the fact that all
agents observe the same sample can be viewed as noise being originated by the
observed agent.}.
\end{itemize}
Regarding the difference in time units between the models,
since interactions occur in parallel in the \ppull model, one round in that
model should  informally be thought of as roughly $n$ time steps in the \spull or \rpull
model.

\subsection{Liberal assumptions}
\label{sec:liberal}
As mentioned, we shall assume that agents have abilities that surpass their
realistic ones. These assumption not only increases the generality of our
lower bounds, but also simplifies their proofs. Specifically, the following
liberal assumptions are considered.

\begin{itemize}[noitemsep,topsep=0pt]
\item {\bf Unique identities.}
    Each agent is equipped with a unique identity $id(v)\in\{1,2,\ldots,n\}$, that
    is, for every two agents $u$ and $v$, we have $id(u)\neq id(v)$. Moreover,
    whenever an agent $u$ observes some agent $v$, we assume that $u$ can
    infer the identity of $v$. In other words, we provide agents with the ability
    to reliably distinguish between different agents at no cost. 
\item {\bf Unlimited internal computational power.} 
    We allow agents to have unlimited
    computational abilities including infinite memory capacity. Therefore, agents
    can potentially perform arbitrarily complex computations based on their
    knowledge (and their $id$).
\item {\bf Complete knowledge of the system.} 
    Informally, we assume that agents have access to the complete description
    of the system except for who are the sources and what is their opinion.
    More formally, we assume that each agent has access to:
    \begin{itemize}[noitemsep,topsep=0pt]
    \item the  neutral initial configuration $\bx^{(0)}$,
    \item all the systems parameters, including the number of agents $n$, the
        noise parameter $\delta$, the number of sources $s$, and the
        distribution $f_{source-state}$ governing the update the states of
        sources in the third stage of the charged initial configuration. 
    \end{itemize}
\item {\bf Full synchronization.}  
    We assume that all agents are equipped with clocks that can count time
    steps (in \spull or \rpull) or rounds (in \ppullk). The clocks are
    synchronized, ticking at the same pace, and initialized to 0 at the
    beginning of the execution. This means, in particular, that if they wish,
    the agents can actually share a notion of time that is incremented at each
    time step. 
\item {\bf Shared randomness.}  
    We assume that algorithms can be randomized. That is, to determine the next
    action, agents can internally toss coins and base their decision on the
    outcome of these coin tosses. Being liberal, we shall assume that
    randomness is shared in the following sense. At the outset, an arbitrarily
    long  sequence $r$ of random bits is generated and the very same sequence
    $r$ is written in each agent's memory before the protocol execution starts.
    Each agent can then deterministically choose (depending on its  state)
    which random bits in $r$ to use as the outcome of its own random bits. This
    implies that, for example, two agents can possibly make use of the very
    same random bits or merely observe the outcome of the random bits used by
    the other agents. Note that the above implies that, conditioning on an
    agent $u$ being a non-source agent, all the random bits used by $u$ during
    the execution are accessible to all other agents.  
 \item {\bf Coordinated sources.} 
    Even though non-source agents do not know who the sources are, we assume
    that sources do know who are the other sources. This means, in particular,
    that the sources can coordinate their actions. 
\end{itemize}

 \subsection{Considered algorithms and solution concept}
Upon observation, each agent can alter its internal state (and in particular,
its message to be seen by others) as well as its opinion.
The strategy in which agents update these variables is called
``algorithm''.  As mentioned, algorithms can be randomized, that is, to
determine the next action, agents can use the outcome of coin tosses in the
sequence $r$ (see \emph{Shared randomness} in Section \ref{sec:liberal}). Overall,
the action of an agent $u$ at time $t$ depends on:
 \begin{enumerate}[noitemsep,topsep=0pt]
    \item the initial state of $u$ in the charged initial configuration  (including the identity of $u$ and whether or not it is a source),
        \item the initial knowledge of $u$ (including the system's parameters and neutral configuration),

        \item the time step $t$, and the list of its observations (history) up to time $t-1$, denoted $x_u^{(<t)}$,
    \item the sequence of random bits $r$.
\end{enumerate}

\subsection{Convergence and time complexity}
At any time, the opinion of an agent can be viewed as a binary \emph{guess}
function that  is used to express its most knowledgeable guess of the correct opinion. 
The agents aim to minimize the probability that they fail to guess this opinion.
In this context, it can be shown that the optimal guessing function is
deterministic. 
\begin{definition}
    \label{def:convergence}
    We say that \emph{convergence} has been achieved if one can specify a
    particular non-source agent $v$, for which is it guaranteed that its
    opinion is the correct opinion with probability at least $2/3$. The \emph{time
    complexity} is the number of time steps (respectively, rounds) required to
    achieve convergence. 
\end{definition}

We remark that the latter definition encompasses all three models  
considered.

\begin{remark}[Different sampling rates of sources]
    We consider sources as agents in the population but remark that they can also
    be thought of as representing the environment. In this case, one may consider a
    different rate for sampling a source (environment) vs. sampling a typical
    agent. For example, the probability to observe any given source (or
    environment) may be $x$ times more than the probability to observe any given
    non-source agent. This scenario can also be captured by a slight adaptation
    of our analysis. When $x$ is an integer, we can alternatively obtain such
    a generalization by considering
    additional \emph{artificial} sources in the system. Specifically, we replace
    each source $u_i$ with a set of sources $U_i$ consisting of $x$ sources that
    coordinate their actions and behave identically,
    simulating the original behavior of $u_i$. 
    (Recall that we assume
    that sources know who are the other sources and can coordinate their
    actions.)
    Since the number of sources
    increases by a multiplicative factor of $x$, our lower bounds (see Theorem
    \ref{thm:main} and Corollary \ref{cor:informal_reliable_source}) decrease by a
    multiplicative factor of $x^2$. 
\end{remark}

\section{The lower bounds}
\label{sec:fixed_coin} 
Throughout this section we consider  $\delta<1/2$, such that $\frac{(1-2\delta)}{\delta sn} \leq \frac
    1{10}$. Our goal in this section is to prove the following result.
\begin{theorem}
    \label{thm:main}
    Assume that the relaxed $\delta$-uniform noise criterion is satisfied.
        \begin{itemize}
        \item Let $k$ be an integer. Any rumor spreading protocol on the \ppullk
            model cannot converge in fewer rounds  than $
                \Omega\left(\frac{n\delta}{k\bias^2(1-2\delta)^2}\right).$
        \item Consider either the  $\spull$ or the \rpull model. Any rumor
            spreading protocol cannot converges in fewer rounds  than $
                \Omega\left(\frac{n^2\delta}{\bias^2(1-2\delta)^2}\right).$
    \end{itemize}
\end{theorem}
To prove the theorem, we first prove (in Section \ref{sec:reduction})  that an efficient  rumor spreading algorithm in either the noisy \spull model or the \ppullk model 
can be used to construct an efficient algorithm in the \rpull  model. The resulted algorithm has the same time
complexity as the original one in the context of \spull and adds a multiplicative factor
of $kn$ in the context of \ppullk. 

We then show how to relate the rumor spreading problem in \rpull to a
statistical inference test (Section \ref{sec:hypothesis}). A lower bound on the
latter setting is then achieved by adapting techniques from mathematical
statistics (Section \ref{sec:aa}).

\subsection{Reducing to the \rpull Model}
\label{sec:reduction}

The following lemma establishes a formal
relation between the convergence times of the models we consider. 
We assume all models are subject to the same noise distribution. 
\begin{lemma} 
    \label{lem:ppullsimul}
            Any protocol operating in \spull can be simulated by a protocol operating in \rpull
            with the same time complexity.
            Moreover, for any integer $1\leq k\leq n$, any protocol $\mathcal P$ operating in \ppullk can be
            simulated by a protocol operating in \rpull with a time complexity that
            is $kn$ times that of $\mathcal P$ in \ppullk. 
\end{lemma}
\begin{proof}
    Let us first show how to  simulate a time step of \spull in the\\\rpull model. 
    Recall that in \rpull, in each time step, all agents receive the same observation sampled
    u.a.r from the population. Upon drawing such an observation, all
    agents use their shared randomness to generate a (shared) uniform
    random integer $X$ between $1$ and $n$. Then, the agent whose unique identity
    corresponds to $X$ is the one processing the observation, while all other
    agents ignore it.  This reduces the situation to a scenario in \spull, and the agents can safely execute the original algorithm designed for that model. 
        
    As for simulating a time step of \ppullk in  \rpull, agents
    divide time steps in the latter model into \emph{rounds}, each composing of precisely $kn$ time steps. 
    Recall that the model assumes that agents share clocks that start
    when the execution starts and tick at each time step. This implies that
    the agents  can agree on the division of time into rounds, and can
    further agree on the round number. For $1\leq i\leq
    kn$, during the $i$-th step of each round, only the agent whose identity is ($i$ mod $n$)$+1$
    receives\footnote{Receiving the observation doesn't imply that the agent
        processes this observation. In fact, it will store it in its memory until
        the round is completed, and process it only then.} 
    the observation, while all other agents ignore it. This ensures that when a
    round is completed in the \rpull model, each agent receives precisely $k$
    independent uniform samples as it would in a round of \ppullk. Therefore,
    at the end of each round $j\in \Natural^+$ in the \rpull model, all agents
    can safely execute their actions in the $j$'th round of the original
    protocol designed for \ppullk. This draws a precise bijection from rounds
    in \ppullk and rounds in \rpull. The multiplicative overhead of $kn$ simply
    follows from the fact that each round in \rpull consists of $kn$ time
    steps. 
\end{proof}

Thanks to Lemma \ref{lem:ppullsimul}, 
Theorem \ref{thm:main} directly follows from the next 
theorem.

\begin{theorem}\label{thm:main_rpull}
    Consider the  $\rpull$ model and assume that the relaxed $\delta$-uniform
    noise criterion is satisfied.     Any rumor
            spreading protocol cannot converges in fewer time steps  than $\Omega\left(\frac{n^2\delta}{\bias^2(1-2\delta)^{2}}\right)~.$
\end{theorem}

The remaining of the section is dedicated to proving  Theorem \ref{thm:main_rpull}. 
Towards achieving this, we view the task of guessing the correct opinion in the  $\rpull$ model, given access to noisy samples, within the 
more general framework of distinguishing between two types of stochastic
processes which obey some specific assumptions. 

\subsection{Rumor Spreading and hypothesis testing}
\label{sec:hypothesis}


To establish the desired lower bound, we next show how the rumor spreading
problem in the broadcast-\pull model relates to a statistical inference test.
That is, from the perspective of a given agent, the rumor spreading problem can
be understood as the following: Based on a sequence of noisy observations, the
agent should be able to tell whether the correct opinion is $0$ or $1$. We
formulate this problem as a specific task of distinguishing between two random
processes,  one originated by running the protocol assuming  the correct
opinion is 0 and the other assuming it is 1. 

One of the main  difficulties lies in the stochastic dependencies affecting these processes. 
In general, at different time steps, they do not consist of independent
draws of a given random variable. In other words, the law of an observation
not only depends on the correct opinion, on the initial configuration and on
the underlying randomness used by agents, but also on the previous noisy
observation samples and (consequently) on the messages agents themselves choose
to display on that round.
An intuitive version of this problem is the task of distinguishing between two
(multi-valued) biased coins, whose bias changes according to the previous
outcomes of tossing them ({\em e.g.,} due to wear). 
%
Following such intuition, we define the following general class of {\em Adaptive Coin
Distinguishing Tasks}, for short $\AVDT$. 

\begin{definition}[$\AVDT$] 
    \label{def:acdt}
    A {\em distinguisher}  is presented with a
    sequence of observations taken from a coin of type $\type$ where
    $\type\in\{0,1\}$. The type $\type$ is initially set to $0$ or $1$ with
    probability $1/2$ (independently of everything else). The goal of the
    distinguisher is to determine the type $\type$, based on the observations.
    More specifically, for a given time step $t$, denote the sequence of previous observations (up to, and including, time $t-1$) by 
    $
    x^{(< t)}=(x^{(1)},\dots,x^{(t-1)}).
    $
    At each time $t$, given the type $\type \in
    \{0,1\}$ and the history of previous observations
    $x^{(< t)}$, the distinguisher receives
    an observation $X^{(t)}_\type\in \Sigma$, which has law\footnote{We follow the common practice to use uppercase letters
    to denote random variables and lowercase letter to denote a particular
    realisation, e.g., $\Seqobs$ for the sequence of observations up
    to time $t$, and $\seqobs$ for a corresponding realization.}
        $\Pr(X_{\type}^{(t)} = m \mid x^{(< t)}). 
        $
\end{definition}
We next introduce, for each $m\in\Sigma$, the parameter
\begin{align*}
    \varepsilon(m, x^{(< t)})= \Pr(X_{1}^{(t)} = m \mid x^{(< t)}) -
    \Pr(X_{0}^{(t)} = m \mid x^{(< t)}).
    \end{align*}
%
Since, at all times $t$, it holds that 
$$
\sum_{m \in \Sigma} \Pr(X_{0}^{(t)} = m
\mid x^{(< t)}) =  \sum_{m \in \Sigma} \Pr(X_{1}^{(t)} = m \mid x^{(< t)}) =
1,$$
 then $\sum_{m \in \Sigma}\varepsilon(m, x^{(< t)})  = 0$. We shall be
interested in the quantity $
    \epsnorm
    :=\sum_{m \in \Sigma}|\varepsilon(m, x^{(< t)})|,$
which corresponds to the $\ell_1$ distance between the distributions
$\Pr(X_{0}^{(t)} = m \mid x^{(< t)})$ and $\Pr(X_{1}^{(t)} = m \mid x^{(< t)})$
given the sequence of previous observations.

\begin{definition}[The bounded family ${\AVDT}(\varepsilon,\delta)$]
    \label{def:bounded}
    We consider a family of instances of
        $\AVDT$, called ${\AVDT}(\varepsilon,\delta)$, governed by parameters $\varepsilon$ and $\delta$. Specifically, 
       this family contains all instances of  ${\AVDT}$ such that for  every $t$, and every history $x^{(<t)}$, we have:
       \begin{itemize}
       \item $\epsnorm \leq \epsilon$, and
       \item  $\forall m\in \Sigma$ such that $\varepsilon(m, x^{(< t)}) \neq 0$, we have $\delta
           \leq  \Pr(X_{\type}^{(t)} = m \mid x^{(< t)})$ for $\type\in \left\{ 0,1 \right\}$.
        \end{itemize}
\end{definition}

In the rest of the current section, we show how Theorem \ref{thm:main_rpull}, that deals with the
$\rpull$ model, follows directly from the next theorem that
concerns the adaptive coin distinguishing task, by setting $\varepsilon=\frac
{2\bias(1-2\delta)} n.$
The actual proof of  Theorem \ref{thm:testcoin} appears in Section \ref{sec:aa}. 

\begin{theorem}
    \label{thm:testcoin}
    Consider any protocol for any instance of $\AVDT(\varepsilon,\delta)$,
    The number of samples
    required to distinguish between a process of type $0$ and a process of
    type $1$ with probability of error less than $\frac{1}{3}$ is at least
$ \frac {\ln 2}9 
            \left( 
            \frac{6(\delta - \epsilon)^3}{\delta^3 - \delta^2\epsilon +3 \delta
            \epsilon^2 -\epsilon^3}
            \right)
            \frac{\delta}{\varepsilon^2}.$
    In particular, if $\frac{\varepsilon}{\delta}<{10}$, then the number of
    necessary samples
    is $\Omega\left(\frac{\delta}{\epsilon^2}\right)$.
\end{theorem}

\subsubsection{Proof of Theorem \ref{thm:main_rpull} assuming Theorem \ref{thm:testcoin}}
    \label{ssec:mappingtocoin}
    %
    Consider a rumor spreading protocol $\cal{P}$ in the \rpull model. 
    Fix a node $u$.
    We first show that running $\cal{P}$ by all agents, the perspective of node $u$ corresponds to a specific instance of
    $\AVDT(\frac {2\bias(1-2\delta)} n,\delta)$ called
    $\Pi(\mathcal{P}, u)$.
    We break down the proof of such correspondence into two claims. 
        
    \paragraph{The $\AVDT$ instance $\Pi(\mathcal{P}, u)$.}
    %
    Recall that we assume that each agent knows the complete neutral initial
    configuration, the number of sources $s$, and the shared  of random bits
    sequence $r$. We avoid writing such parameters as explicit arguments to
    $\Pi(\mathcal{P}, u)$ in order to simplify notation, however, we stress that what follows assumes that these parameters are fixed.
The bounds we show hold for any fixed value of $r$ and hence also when $r$ is randomized.

    Each agent is interested in discriminating between two families of
    charged initial configurations: Those in which the correct opinion is $0$ and
    those in which it is $1$ (each of these possibilities occurs with probability
    $\frac 12$).
    Recall that the correct opinion is determined in the 2nd  stage of the
    charged initial configuration, and is independent form the choice of sources
    (1st stage).
    
    We next consider the perspective of  a generic non-source agent $u$, and define the instance  $\Pi(\mathcal{P}, u)$ as follows. 
    Given the history $x^{(< t)}$, we set
    $\Pr(X_{\type}^{(t)} = m \mid x^{(< t)})$, for $\type\in \left\{ 0,1 \right\}$,
    to be equal to the probability that $u$ observes message $m\in\Sigma$ at time
    step $t$ of the execution $\cal{P}$. 
    For clarity's sake, we remark that the latter probability is conditional
    on: the history of observations being $x^{(< t)}$,
        ,the sequence of random bits $r$,
        ,the correct opinion being $\type\in \{0,1\}$,
        ,the neutral initial configuration,
        ,the identity of $u$,
        ,the algorithm $\mathcal P$, and
        the system's parameters (including  the distribution $f_{source-state}$ and the number of sources $s$).
    %
    %
    \begin{claim}
    \label{claim:subclaim2}
        Let $\cal{P}$  be a correct protocol for the rumor spreading problem in
        \rpull and let $u$ be an agent  for which the protocol is guaranteed to
        produce the correct opinion with probability at least $p$ by some
        time $T$ (if one exists), for any fixed constant $p\in (0,1)$. 
        Then $\Pi({\cal{P}},u)$ can be solved in time $T$ with correctness
        being guaranteed with probability at least $p$. 
    \end{claim} 
    \begin{proof}
        Conditioning on $\eta\in\{0,1\}$ and on the random seed $r$, 
        the distribution of observations in the $\Pi(\mathcal{P}, u)$ instance
        follows precisely the distribution of observations as perceived from
        the perspective of $u$ in \rpull. 
        Hence, if the protocol $\cal{P}$ at $u$ terminates  with output $j\in
        \{0,1\}$ at round $T$, after the $T$-th observation in $\Pi({\cal{P}},u)$ we
        can set $\Pi({\cal{P}},u)$'s output to $j$ as well. Given that the two
        stochastic processes have the same law, the correctness guarantees are
        the same. 
    \end{proof}

    \begin{lemma} 
        \label{lem:subclaim1}
        $\Pi({{\cal{P}},u})\in \AVDT\left(\frac{2(1-2\delta)s}{n},\delta\right)$.
    \end{lemma}      
    \begin{proof}
        Since the noise in \rpull flips each message $m\in\Sigma$ into any
        $m'\in \Sigma'$ with probability at least $\delta$, regardless of the
        previous history  and of $\type \in \{0,1\}$, at all times $t$, if
        $m\in \Sigma'$ then $\Pr(X_{\type}^{(t)} = m \mid x^{(< t)}) \geq \delta.$
        %
        %
        Consider a message $m\in  \Sigma\setminus \Sigma'$ (if such a message exists). By
        definition, such a  message could only be received  by observing a
        non-source agent. But given the same history $x^{(<t)}$, the same sequence of random bits $r$, and the same
        initial knowledge, the behavior of a
        non-source agent is the same, no matter what is the correct opinion
        $\eta$. Hence, for $m\in \Sigma\setminus \Sigma'$ we have
        $\Pr(X_{0}^{(t)} = m \mid x^{(< t)}) = \Pr(X_{1}^{(t)} = m \mid x^{(<
        t)})$, or in other words, 
        $m\in  \Sigma\setminus \Sigma' \implies \varepsilon(m, x^{(< t)}) = 0.$
        
        It remains to show that $\epsnorm \leq \frac{2(1-2\delta)s}{n}$. 
        Let us consider two  executions of the rumor spreading
        protocol, with the same neutral initial configuration, same shared
        sequence of random bits $r$, same set of sources, except that in the
        first the correct opinion is $0$ while in the other it is $1$. 
        Let us condition on the history of observations $x^{(<t)}$ being the
        same in both processes. 
        %
        %
        As mentioned,  given the same history $x^{(<t)}$, the behavior of a non-source
        agent is the same, regardless of the correct opinion $\eta$.  It
        follows that the difference in the 
        probability of observing any given message is only due to the event that a source is observed.
        Recall that the number of  sources is $s$. Therefore, the probability of observing a source is  $s/n$, and we may write as a first approximation 
        $\varepsilon(m, x^{(< t)}) \leq \bias/n$
        . However, we can be more precise. In fact, $\varepsilon(m, x^{(< t)})$ is slightly smaller than $\bias/n$, because the noise can still affect the message of a source.
        We may interpret $\varepsilon(m, x^{(< t)})$ as the following
        difference. For a source $v\in S$, let $m_\eta^{v}$ be the message of
        $u$ assuming the given history  $x^{(< t)}$ and that $v$ is of type $\type \in \{0,1\}$ (the
        message $m_\eta^{v}$  is deterministically determined given the
        sequence $r$ of random bits, the neutral initial configuration, the parameters of the system, and the
        identity of $v$). Let $\alpha_{m',m}$ be the probability that the noise
        transforms a message $m'$ into a message $m$. Then 
        $
            \varepsilon(m, x^{(< t)}) 
            = \frac 1 n \sum_{v \in S} (\alpha_{m_1^{v},m} -
            \alpha_{m_0^{v},m}),
            $
        and 
        \begin{equation}
            \epsnorm
            =\sum_{m\in \Sigma} |\varepsilon(m, x^{(< t)})|
            \leq \frac 1 n\sum_{m\in \Sigma}\sum_{v \in S} |\alpha_{m_1^{v},m}
            - \alpha_{m_0^{v},m}|.
            \label{eq:epsnorm}
        \end{equation}
        By the definition of ${\AVDT}(\varepsilon,\delta)$, it follows that
        either $\alpha_{m_1^{v},m} = \alpha_{m_0^{v},m}$ (if $\varepsilon(m,
        x^{(< t)})=0$) or $\delta \leq\alpha_{m_1^{v},m}, \alpha_{m_0^{v},m}
        \leq 1-\delta$ (if $ \varepsilon(m, x^{(< t)})\neq 0$). Thus, to bound
        the right hand side in \eqref{eq:epsnorm}, we can use the following
        claim.
        \begin{claim}\label{claim:tv}
            Let $P$ and $Q$ be two distributions over a universe $\Sigma$ such that for any element $m\in \Sigma$,
            $\delta \leq P(m), Q(m) \leq 1-\delta.$
            Then $\sum_{m\in \Sigma} \lvert P(m) - Q(m) \rvert \leq 2 (1-2\delta).$
        \end{claim}

        \begin{proof}[Proof of Claim \ref{claim:tv}]\label{sec:app-tv}
            Let $\Sigma_+ := \{m: P(m)>Q(m) \}$. We may write    
            \begin{align*}        
            \sum_{m\in \Sigma} \lvert P(m) - Q(m) \rvert &= \sum_{m\in \Sigma_+} (P(m) - Q(m)) + \sum_{m\in setminus \Sigma_+} (Q(m) - P(m))\\
                &= P(\Sigma_+) - Q(\Sigma_+) + Q(\Sigma \setminus \Sigma_+) - P(\Sigma \setminus \Sigma_+)\\
                &= 2( P(\Sigma_+) - Q(\Sigma_+)),
            \end{align*}
            where in the last line we used the fact that 
            $Q(\Sigma \setminus \Sigma_+) - P(\Sigma \setminus \Sigma_+)=
            1-Q(\Sigma_+) - 1 + P(\Sigma_+) = P(\Sigma_+) - Q(\Sigma_+).$
            We now distinguish two cases. 
            \textbf{Case $1$.}If $\Sigma_+$ is a singleton, $\Sigma_+ = \{m^*\}$, then 
            $P(\Sigma_+)-Q(\Sigma_+)= P(m^*)-Q(m^*) \leq 1-2\delta$, by assumption.
            \textbf{Case $2$.} Otherwise, $\lvert \Sigma_+ \rvert \geq 2$ and 
$2 \sum_{m\in \Sigma_+} (P(m) - Q(m)) \leq 2- 2\sum_{m \in \Sigma_+}Q(m) \leq 2(1-\delta\lvert\Sigma_+ \rvert) \leq 2(1-2\delta)$,
                using the fact that for any $m$, $Q(m) \geq \delta$, and the
                fact that $P$ is a probability measure.
              This completes the proof of Claim \ref{claim:tv}.
        \end{proof}

        \bigskip
        Applying Claim \ref{claim:tv} for a fixed $v\in S$ to distributions $(\alpha_{m_0^{v},m})_m$ and 
        $(\alpha_{m_1^{v},m})_m$, we obtain
        \[
            \frac 1 n\sum_{m\in \Sigma}\sum_{v\in S} |\alpha_{m_1^{v},m} - \alpha_{m_0^{v},m}|
            \leq \frac 1 n 2 \sum_{v\in S} (1-2\delta)
            \leq \frac{2(1-2\delta)\bias}{n}.
        \] 
        Hence, we have $\Pi({\cal{P}})\in
        \AVDT\left(\frac{2(1-2\delta)\bias}{n},\delta\right)$, establishing
        Lemma \ref{lem:subclaim1}.
    \end{proof}  

    Thanks to Claims \ref{claim:subclaim2} and Lemma \ref{lem:subclaim1}, Theorem
    \ref{thm:main_rpull} regarding the $\rpull$ model becomes a direct
    consequence of Theorem \ref{thm:testcoin} on the adaptive coin
    distinguishing task, taking $\varepsilon=\frac {2(1-2\delta)\bias} n~.$
    More precisely, the assumption $\frac{(1-2\delta)}{\delta sn}\leq c$ for
    some small constant $c$, ensures that $\frac{\varepsilon}{\delta} \leq c$
    as required by Theorem \ref{thm:testcoin}. The lower bound
    $\Omega\left(\frac{\varepsilon^2}{\delta}\right)$ corresponds to $\Omega\left(\frac{n^2\delta}{(1-2\delta)^{2}\bias^2}\right)~.$ This
    concludes the proof of Theorem \ref{thm:main_rpull}. 
    \qed
    
    \medskip
    To establish our
    results it remains to prove Theorem \ref{thm:testcoin}. 

\subsection{Proof of Theorem \ref{thm:testcoin} }
\label{sec:aa}

We start by recalling some facts from Hypothesis Testing. First let us
recall two standard notions of (pseudo) distances between probability
distributions. Given two discrete distributions $P_0,P_1$ over a probability space
$\Omega$ with the same support%
\footnote{The assumption that the support is the same is not necessary but it
    is sufficient for our purposes, and is thus made for
    simplicity's sake.
    },
the {\em total variation distance} is defined as $
    TV(P_0,P_1) :=  \frac 12 \sum_{x \in \Omega} \lvert P_0(x) - P_1(x) \rvert,$
and the  Kullback-Leibler divergence $\kl(P_0,P_1)$ is
defined\footnote{
    We use the notation $\log(\cdot)$  to denote the base 2 logarithms, i.e.,
    $\log_2(\cdot)$ and for a probability distribution $P$, use the notation $P(x)$
    as a short for $P(X=x)$.} 
    as $
    \kl(P_0,P_1) := \sum_{x\in \Omega} P_0(x) \log \frac{P_1(x)}{P_0(x)}.$

The following lemma shows that, when trying to discriminate between
distributions $P_0,P_1$, the total variation relates to the smallest error probability we
can hope for.

\begin{lemma}[{{Neyman-Pearson \cite[Lemma $5.3$ and Proposition $5.4$]{rig15}}}]
    \label{lem:neyman}
    Let $P_0,P_1$ be two distributions. Let $X$ be a random variable
    of law either $P_0$ or $P_1$. Consider a (possibly probabilistic) mapping
    $f :\Omega\rightarrow\{0,1\}$ that attempts to ``guess'' whether
    the observation $X$ was drawn from $P_0$ (in which case
    it outputs $0$) or from $P_1$ (in which case it outputs $1$). 
    Then, we have the following lower bound,
    \begin{equation}
        \text{\ensuremath{{\textstyle \Pr_{0}}}}\left(f(X)=1\right) 
        + \text{\ensuremath{{\textstyle \Pr_{1}}}}\left(f(X)=0\right) 
            \geq 1-TV(P_0,P_1).
    \end{equation}
\end{lemma}
The total variation is related to the $\kl$ divergence by the following inequality.
\begin{lemma}[{{Pinsker \cite[Lemma $5.8$]{rig15}}}]
    \label{lem:pinsker}
    For any two distributions $P_0,P_1$,
    \begin{equation}
        TV(P_0,P_1) \leq \sqrt{ KL\left(P_0,P_1\right)}.
    \end{equation}
\end{lemma}

We are now ready to prove the theorem. 

\begin{proof}[Proof of Theorem \ref{thm:testcoin}]
    Let us define
    $\text{\ensuremath{{\textstyle \Pr_{\eta}}}}\left(\ \cdot\ \right)=\Pr\left(\
    \cdot\cond\text{``correct distribution is $\eta$''}\right)$
    for $\eta\in \left\{ 0,1 \right\}$.
    We denote $P^{(\leq t)}_\type$, $\type \in \left\{ 0,1 \right\}$, the two
    possible distributions of $\Seqobs$. 
    We refer to  $P_0^{(\leq t)}$ as the distribution of
    \emph{type $0$} and to $P_1^{(\leq t)}$ as the distribution of \emph{type $1$}. Furthermore,
    we define the \emph{correct type} of a sequence of observations $\Seqobs$ to be 
    $0$ if the observations are sampled from $P^{(\leq t)}_0$, and to be $1$ if they are
    sampled from $P^{(\leq t)}_1$. 

    After
    $t$ observations $\seqobs = (\obs 1, \dots, \obs t)$ we have to decide
    whether the distribution is of type $0$ or $1$. 
    Our goal is to
    maximize the probability of guessing the type of the distribution, observing $\Seqobs$,
    which means that we want to minimize
\begin{equation}
        \label{eq:total}
        f
        =\sum_{\eta\in\{ 0,1 \}}
       \Pr_{\eta}\left(f(\Seqobs)=1-\eta\right)
         \Pr\left(\text{``correct type is $\eta$''} \right). 
\end{equation}
    Recall that the correct type is either $0$ or $1$ with probability $\frac 12$. 
    Thus, the error probability described in \eqref{eq:total} becomes
    \begin{equation}
         \frac 12 \Pr_{0}\left(f(\Seqobs)=1\right)
            + \frac 12 
             \Pr_{1}
            \left(f(\Seqobs)=0\right).
            \label{eq:after_prior}
    \end{equation}

    By combining Lemmas \ref{lem:neyman} and  \ref{lem:pinsker} with $X =
    \Seqobs$ and $P_\type= P_{\type}^{(\leq t)}$ for $\type=0,1$,
     we get the following Theorem. Although for convenience we think of $f$ as a deterministic
     function, it could in principle be randomized.
    \begin{theorem}
        \label{thm:aux}
        Let $f$  be any guess function. Then
        \begin{align}
         \Pr_{0}\left(f
            t(\Seqobs)=1\right) + \Pr_{1}\left(f(\Seqobs)=0\right) \geq 1- \sqrt{
            KL\left(P_0^{(\leq t)}, P_1^{(\leq t)}\right)}.
            \label{eq:lowerKL}
        \end{align}
    \end{theorem}
Theorem \ref{thm:aux} implies that 
 for the probability of error to be small, it must be the case that the term $KL\left(P_0^{(\leq t)}, P_1^{(\leq t)}\right)$ is large. 
 Our next goal is therefore to show that in order to make this term large, $t$ must be large. 
    
    Note that $P_\type^{(\leq T)}$ for $\type \in \{0, 1\}$ cannot be written
    as the mere product of the marginal distributions of the $X^{(t)}$s, since the
    observations at different times may not necessarily be independent. Nevertheless, we can still express the
    term $
      \kl(P_0^{(\leq T)},P_1^{(\leq T)})$
       as a sum, using the Chain Rule for $\kl$ divergence\footnote{See Lemma $3$ in
      \url{http://homes.cs.washington.edu/anuprao/pubs/CSE533Autumn2010/lecture3.pdf}.}. 
      It yields
    \begin{align}
        \kl(P_0^{(\leq T)},P_1^{(\leq T)}) &= \sum_{t\leq T}
        \kl(P_0(x^{(t)}\mid x^{(<t)}),P_1(x^{(t)} \mid x^{(<t)}))
        \label{eq:KL}\\
        &:= \sum_{x^{(<t)} \in \Sigma^{t-1}} P_0(x^{(<t)})
            \sum_{x^{(t)}\in \Sigma}
                P_0(x^{(t)}\mid x^{(<t)}) \log \frac{P_0(x^{(t)}\mid
                x^{(<t)})}{P_1(x^{(t)} \mid x^{(<t)})}.\nonumber\\
        &= \sum_{x^{(<t)} \in \Sigma^{t-1}} P_0(x^{(<t)})
            \sum_{ m \in \Sigma}
                P_0( X_{0}^{(t)} = m \mid x^{(<t)}) \log \frac{P( X_{0}^{(t)} = m \mid
                x^{(<t)})}{P( X_{1}^{(t)} = m \mid x^{(<t)})}.
        \label{eq:condKL}
    \end{align}
    Since we are considering an instance of  
    $\AVDT\left( \epsilon,\delta\right)$, we have 
    \begin{itemize}
        \item $\epsnorm =\sum_{m \in \Sigma}|\varepsilon(m, x^{(< t)})| \leq \epsilon$, and         
        \item for every $m\in \Sigma$ such that $\varepsilon(m, x^{(< t)}) \neq 0$, it holds that $\delta
           \leq  \Pr_\type( X_{0}^{(t)} = m \mid x^{(< t)})$ for $\type\in \left\{ 0,1 \right\}$.
    \end{itemize}
    We make use of the previous facts to upper bound the \kl~divergence terms in the
    right hand side of \eqref{eq:condKL}, as follows.
    
    %
    \begin{align}
        &\kl(P_0(x^{(t)}\mid x^{(<t)}),P_1(x^{(t)} \mid x^{(<t)}))\\  
        &= \sum_{x^{(<t)} \in \Sigma^{t-1}} P_0(x^{(<t)}) \sum_{m\in \Sigma}
            \left( \Pr(X_{0}^{(t)} = m \mid x^{(< t)}) 
                \log \frac{\Pr(X_{0}^{(t)} = m \mid x^{(< t)})}
                        {\Pr(X_{0}^{(t)} = m \mid x^{(< t)})+ \varepsilon(m, x^{(< t)})} \right)\\
        &= - \sum_{x^{(<t)}} P_0(x^{(<t)}) \sum_{m\in \Sigma}
            \left( \Pr(X_{0}^{(t)} = m \mid x^{(< t)}) 
                \log\left(1+\frac{\varepsilon(m, x^{(< t)})}{\Pr(X_{0}^{(t)} = m \mid x^{(< t)})}\right) \right).
        \label{eq:innerlog}
    \end{align}
    Recall that we assume $
        \frac{\varepsilon(m, x^{(< t)})}{\Pr(X_{0}^{(t)} = m \mid x^{(<
        t)})} 
        \leq \frac{\varepsilon(m, x^{(< t)})}{\delta}
        \leq \frac{\varepsilon}{\delta}.$
    We make use of the following claim, which follows from the Taylor expansion
    of $\log(1+u)$ around $0$. 
    \begin{claim}\label{claim:taylor}
        Let $x \in [-a , a]$ for some $a\in (0,1)$. Then
            $ \lvert \log(1+x) - x + x^2 / 2 \rvert \leq \frac{x^3}{3(1-a)^3}$.
    \end{claim}
   Using Claim \ref{claim:taylor} with $a=\frac \epsilon\delta$, we can bound 
    the inner sum appearing in
    \eqref{eq:innerlog} from above and below with 
    \begin{equation}
        \frac{1}{\ln 2} \sum_{m \in \Sigma}
            \left({\varepsilon(m, x^{(< t)})} -
            \frac{1}{2}\frac{(\varepsilon(m, x^{(< t)}))^2}{\Pr(X_{0}^{(t)} = m \mid x^{(< t)})} \pm
        \frac{\delta^3}{3(\delta - \epsilon)^3} \left( \frac{(\varepsilon(m, x^{(< t)}))^3}{\Pr(X_{0}^{(t)} = m
        \mid x^{(< t)})^2} \right) \right).
        \label{eq:taylor}
    \end{equation}

    Since $\sum_{m} \lvert \varepsilon(m, x^{(< t)}) \rvert \leq \varepsilon$,
    we also have that $\sum_m{(\varepsilon(m, x^{(< t)}))^2} \leq \epsilon^2$.
    The latter bound, together with the fact that
    $\Pr(X_{0}^{(t)} = \tilde m \mid x^{(< t)})\geq \delta$
    for any $\tilde m\in \Sigma$ such that $\varepsilon(\tilde m, x^{(<
    t)})\neq 0$,
    implies
    \begin{equation}\label{eq:flower}
        \sum_m\frac{(\varepsilon(m, x^{(< t)}))^2}{\Pr(X_{0}^{(t)} = m \mid x^{(< t)})}
        \leq \frac{\varepsilon^2}{\delta}~.
    \end{equation}
    Finally, we can similarly bound the term $\sum_{m \in \Sigma}
    \left({(\varepsilon(m, x^{(< t)}))^3}/{\Pr(X_{0}^{(t)} = m \mid x^{(<
    t)})^2}\right)$ with
    \begin{equation}\label{eq:power}
        \sum_{m \in \Sigma}
        \left({(\varepsilon(m, x^{(< t)}))^3}/{\Pr(X_{0}^{(t)} = m \mid x^{(<t)})^2}\right)
        \leq  \frac{\varepsilon^3}{\delta^2}.
    \end{equation}

    Recall that $\sum_{m}{\varepsilon(m, x^{(< t)})} =0$, thus the first term
    in \eqref{eq:taylor} disappears. 
    Hence, 
    substituting the bounds \eqref{eq:flower} and \eqref{eq:power} in \eqref{eq:taylor}, 
    we have 
    \begin{align}
        \left|
            \log\left(1+\frac{\varepsilon(m, x^{(< t)})}{\Pr(X_{0}^{(t)} = m \mid x^{(< t)})}\right)
        \right|
        &\leq \frac{1}{\ln 2} 
            \left( \frac{1}{2}\frac{\varepsilon^2}{\delta}  +
            \frac{\delta\epsilon^3}{3(\delta - \epsilon)^3}
            \right)\nonumber
            \\
        &
        \leq    
        \frac{1}{\ln 2} 
            \left( \frac{1}{2} +
            \frac{\delta^2\epsilon}{3(\delta - \epsilon)^3}
            \right)
            \frac{\varepsilon^2}{\delta}.
        \label{eq:upper}
    \end{align}
    %
    %
    %
    If we define the right hand side \eqref{eq:upper} to be $W(\epsilon, \delta)$ and we
    %
        substitute the previous bound in \eqref{eq:innerlog}, we get
    \begin{equation}
        \kl(P_0(x^{(t)}\mid x^{(<t)}),P_1(x^{(t)} \mid x^{(<t)})) 
        \leq 
        W(\epsilon, \delta),
        \label{eq:bouding}
    \end{equation}
    and combining the previous bound with \eqref{eq:KL}, we can finally
    conclude that for any integer $T$, we have $\kl(P_0^{(\leq T)},P_1^{(\leq T)})  \leq T \cdot W(\epsilon, \delta).$
    Thus, from Theorem \ref{thm:aux} and the latter bound, it follows that the
    error under a uniform prior of the source type, as defined in
    \eqref{eq:after_prior}, is at least
    \begin{align}
       \frac 12
            \text{\ensuremath{{\textstyle \Pr_{0}}}}\left(f
                t(\Seqobs)=1\right) + \frac 12\text{\ensuremath{{\textstyle
                \Pr_{1}}}}\left(f(\Seqobs)=0\right) 
            &\geq \frac 12 - \frac 12
                \sqrt{\kl(P_0^{(\leq T)},P_1^{(\leq T)})}\\
            &\geq \frac 12 - \frac 12 \sqrt{T \cdot W(\epsilon, \delta)}.
        \label{eq:anapple}
    \end{align}
    Hence, the
    number of samples $T$ needs to be greater than 
$
        \frac 19 \frac{1}{W(\epsilon, \delta)}
        = 
        \frac {\ln 2}9 
            \left( 
            \frac{6(\delta - \epsilon)^3}{\delta^3 - \delta^2\epsilon +3 \delta
            \epsilon^2 -\epsilon^3}
            \right)
            \frac{\delta}{\varepsilon^2}$
    to allow the possibility that the error be
    less than $1/3$. 

    In particular, if we assume that ${10}\epsilon< \delta$, then we can bound
$\frac{\delta^2\epsilon}{3(\delta - \epsilon)^3} \leq \frac{\delta^3}{10}
        \cdot \frac{1}{3(9/10)^3 \delta^3}
        \leq \frac{100}{2187}.$
    It follows that \eqref{eq:upper} can be bounded with $ W(\varepsilon, \delta) \leq \frac{1}{\ln 2}\left(\frac{1}{2} +
        \frac{100}{2187}\right) \leq 0.79 \,,$
    and so $\frac 19
        \frac{1}{W(\epsilon, \delta)} \geq 0.14 \cdot \frac{\delta}{\epsilon^2}
        = \Omega\left(\frac{\delta}{\epsilon^2}\right).$
    This completes the proof of Theorem \ref{thm:testcoin} and hence of
    Theorem \ref{thm:main_rpull}.
\end{proof}

\bibliography{Biblio}

\appendix

\section{Methods}\label{app:exp}
All experimental results presented in this manuscript are re-analysis of data obtained in \emph{Cataglyphis niger} recruitment experiments \cite{Razin}. In short, ants in the entrance chamber of an artificial nest were given access to a tethered food item just outside the nest's entrance (Fig \ref{fig:razin}a). 
The inability of the ants to retrieve the food induced a recruitment process \cite{Razin}. \\The reaction of the ants to this manipulation was filmed and the locations, speeds and interactions of all participating ants were extracted from the resulting videos. \\
\noindent {\bf Calculation of $\delta$.}
To estimate the parameter $\delta$ we used interactions between ants moving at four different
speed ranges (measured in ${cm}/{sec}$), namely, `a': 0-1, `b': 1-5, `c': 5-8, and `d': over 8, and stationary ``receiver''
ants where used. The message alphabet is then assumed to be $\Sigma=\{a,b,c,d\}$. The response of a stationary ant $v$ to the interaction was quantified in terms of her speed after the interaction. 
Assuming equal priors to all messages in $\Sigma$, and given specific
speed of the receiver ant, $v$, the probability that it was the result of a specific message $i\in \Sigma$ was
calculated as $p_i(v)={p(v\mid i)}/{\sum_{k\in \Sigma} p(v\mid k)}$, where ${p(v\mid j)}$ is the probability of responding in speed $v$ after ``observing''  $j$. The probability $\delta (i,j)$ that message $i$
was perceived as message $j$ was then estimated as the weighted sum over the entire
probability distribution measured as a response to $j$: 
$\delta (i,j) = \sum_v p(v\mid j) \cdot p_i(v)$. 
The parameter $\delta$ can then be calculated using 
$\delta =\min\{\delta (i,j)\mid i,j\in\Sigma\}$.

\section{Proof of Corollary \ref{cor:informal_reliable_source}}\label{app:cor_source}
    Let us start with the first item of the corollary, namely the lower bound
    in the $\spull$ model. For any step $t$, let $S(t)$ denote the set of
    sources together with the agents that have directly observed at least one
    of the sources at some point up to time $t$. We have $S=S(0)\subseteq
    S(1)\subseteq S(2)\subseteq \ldots$. The size of the set $S(t)$ is a random variable
    which is expected to grow at a moderate speed. Specifically,   letting $s'=\frac{11}{10}\cdot s \cdot T$, we obtain: 
    \begin{claim}
        \label{claim:chernoff}
       With probability at least $1-n^{-10}$, we have $\lvert S(T)\rvert \leq  s'$.
    \end{claim}
    \begin{proof}[Proof of Claim \ref{claim:chernoff}]
        The variable $S(T)$  may be written as a sum of indicator variables 
        \begin{align*}
            S(T) 
            &= \sum_{i=1}^n \mathds{1}(\text{Agent $i$ observed at least one 
                source before step $t$})\\
            &\leq \sum_{i=1}^n \sum_{r\leq T} \mathds{1}(\text{Agent $i$ observes
                a source on step $r$}).
        \end{align*}
        This last expression is a sum of $n\cdot T$ independent Bernoulli
        variables with parameter $s/n$. In other terms, it is a binomial
        variable with probability $s/n$ and $T \cdot n$ trials. By a standard
        Chernoff bound the probability that it deviates by a multiplicative
        factor $\frac{11}{10}$ from its mean $s\cdot T$ is less than
        $\exp(-\Omega(sT)) \leq  n^{-10}$. The last bound holds because we assume $sT\geq C \log n$ for some large enough constant $C$.
    \end{proof}
    \newcommand{\event}{\mathcal E}
    \newcommand{\prob}{\rho}
    Denote by $\event$ the event that $|S(t)|\leq s'$ for
    every $t\leq T$. Using Claim \ref{claim:chernoff}, we know that
    $\Pr(\event) \geq 1-n^{-10}$.
    Our goal next is to prove that the probability $\prob$ that a given agent
    correctly guesses the correct opinion is low for any given time $t\leq c
    T$, where $c$ is a small constant. For this purpose, we condition on
    the highly likely event $\event$. Removing this conditioning will amount to
    adding a negligible term (of order at most $n^{-10}$) to $\prob.$
   
    In order to bound $\prob$, we would like to invoke Theorem
    \ref{thm:main_rpull} with the number of sources upper bounded by
    $s'$. Let us explain why it applies in this context. To begin with, we may
    adversarially assume (from the perspective of the lower bound) that all agents in $S(t)$ learn the value of the correct
    bit to spread. Thus, they essentially become ``sources'' themselves. In
    this case the number of sources varies with time, but the proof of Theorem
    \ref{thm:main_rpull}
    can easily be shown to cover this case as long as $s$ (i.e.,
    $s'$ here) is an upper bound on the number of sources at all times. We can
    therefore safely apply Theorem \ref{thm:main_rpull} with $s'$. 
    By the choice of $T$, 
    \[
        T = \Theta\left(\frac{n^{2}\delta}{(s')^2(1-2\delta)^2}\right)~.
    \] 
     Hence, we can set $c$ to be a  sufficiently small constant such that for all times $t \leq cT$,    the probability of guessing correctly, even in this adversarial scenario,
    is less than $1/3$. In other words, we have $\prob\leq 1/3$.
    All together, this yields a lower bound of $\Omega(T)$ on the convergence time.
    \medskip
    
    As for the \ppullk model, the argument is similar. 
    After $T'={T}/{k}$ parallel rounds,
    using a similar claim as Claim \ref{claim:chernoff}, we have that with high
    probability, at most
    $\bigO(ks T')$ agents have directly observed one of the $s$ sources by time $T'$. 
    Applying Theorem \ref{thm:main_rpull} with $s'' =
    \bigO(ks T')= \bigO(sT)$ yields a lower bound (in terms of samples in the broadcast model) of
    \begin{equation}
    \Theta\left(\frac{n^{2}\delta}{(s'')^2(1-2\delta)^2}\right) = 
    \Theta\left(\frac{n^{2}\delta}{s^2 T^2(1-2\delta)^2}\right) = \Theta(T).
    \end{equation}
    The last line follows by choice of $T$.
    Hence $T$ is a lower bound on the number of samples, which is attained in $T'$ rounds of \ppullk model.
\end{document}